\documentclass[letterpaper, 10 pt, conference]{ieeeconf}  

\IEEEoverridecommandlockouts                              

\usepackage{graphics} 
\usepackage{amsmath} 
\usepackage{amssymb} 

  \usepackage{amsthm}
\usepackage{mathtools}
\usepackage{algorithm}
\usepackage{algpseudocode}

\renewcommand{\cal}[1]{\mathcal{ #1 }}
\newcommand{\mb}[1]{\mathbf{ #1 }}
\newcommand{\bs}[1]{\boldsymbol{ #1 }}
\newcommand{\bb}[1]{\mathbb{ #1 }}
\newcommand{\norm}[1]{\left\Vert #1 \right\Vert}

\newcommand{\grad}{\nabla}
\newcommand{\intersect}{\cap}
\newcommand{\union}{\cup}

\newcommand{\R}{\bb{R}}

\DeclareMathOperator*{\argmin}{arg\,min}

\newtheorem{theorem}{Theorem}
\newtheorem{corollary}{Corollary}
\newtheorem{lemma}{Lemma}
\theoremstyle{definition}
\newtheorem{definition}{Definition}
\theoremstyle{remark}
\newtheorem{remark}{Remark}
\theoremstyle{definition}
\newtheorem{assumption}{Assumption}
\theoremstyle{definition}
\newtheorem{proposition}{Proposition}

\title{\LARGE \bf A Control Lyapunov Perspective on Episodic Learning\\ via Projection to State Stability}

\author{Andrew J. Taylor$^{1}$, Victor D. Dorobantu$^{1}$, Meera Krishnamoorthy, \\ Hoang M. Le, Yisong Yue, and Aaron D. Ames
\thanks{*This work was supported by Google Brain Robotics and DARPA Award HR00111890035}
\thanks{$^{1}$Both authors contributed equally.}%
\thanks{All authors are with the Department of Computing and Mathematical Sciences, California Institute of Technology, Pasadena, CA 91125, USA {\tt\small ajtaylor@caltech.edu, vdoroban@caltech.edu, mkrishna@caltech.edu, hmle@caltech.edu, yyue@caltech.edu, ames@caltech.edu}}%
}

\begin{document}

\maketitle
\thispagestyle{empty}
\pagestyle{empty}

\begin{abstract}

The goal of this paper is to understand the impact of learning on control synthesis from a Lyapunov function perspective.  In particular, rather than consider uncertainties in the full system dynamics, we employ Control Lyapunov Functions (CLFs) as low-dimensional projections. To understand and characterize the uncertainty that these projected dynamics introduce in the system, we introduce a new notion: Projection to State Stability (PSS). PSS can be viewed as a variant of Input to State Stability defined on projected dynamics, and enables characterizing  robustness of a CLF with respect to the data used to learn system uncertainties. We use PSS to bound uncertainty in affine control, and demonstrate that a practical episodic learning approach can use PSS to characterize uncertainty in the CLF for robust control synthesis.
    
    
    
    
\end{abstract}

\section{INTRODUCTION}
\label{sec:intro}

Properly characterizing uncertainty is a key aspect of robust control \cite{zhou1996robust}.  
With the increasing use of learning for dynamics modelling and control synthesis \cite{aswani2013provably,chowdhary2015bayesian,berkenkamp2016bsafe,dean2018regret,amos2018differentiable,taylor2019episodic,shi2019neural}, it is correspondingly important to develop new tools to reason about the interplay between learning and robust control.

In this paper, we focus on the interplay between learning and robustness for control synthesis using Control Lyapunov Functions (CLFs) \cite{artstein1983stabilization,lin1991universal}.  
%
%
The use of CLFs has seen multiple applications in recent years \cite{ma2017bipedal,galloway2015torque,nguyen2015optimal}, and one of their primary benefits is to  enable control objectives to be represented in a low-dimensional form that can be integrated with optimization methods to yield optimal controllers \cite{ames2013towards}.  This low-dimensional form is also appealing from a learning perspective, as learning is typically more tractable in lower-dimensional spaces \cite{vapnik1999overview,zhou2002covering,taylor2019episodic}.

The practical design of CLFs remains challenging.  In many cases, extensive tuning upon deployment is necessary \cite{ma2017bipedal}, and even with this tuning the system is often not able to track  a desired state or trajectory perfectly.  Other approaches, such as those based on adaptive control \cite{krstic1995control}, can adaptively learn a CLF but are restricted to learning over specific classes of model uncertainty. 

We thus build upon ideas in robust control in order to guarantee performance in the presence of model mis-specification.  The idea of robust CLFs is not new (cf. \cite{freeman1996inverse,freeman2008robust}), but existing  analyses  focus on the full-dimensional state dynamics, which can be burdensome for learning.


In this paper, we make two main contributions.  First, we propose a novel characterization called \textit{Projection to State Stability} (PSS), which is a variant of the well-studied Input to State Stability (ISS) property \cite{sontag1989smooth,sontag1995characterizations,sontag1995characterizations2,wang1996converse,sontag2008input}, but defined on projected dynamics rather than the original state dynamics. 
Like ISS, PSS provides a tool to characterize  tracking error in terms of the magnitude of the disturbance or uncertainty.  Unlike ISS, PSS can  characterize dynamic uncertainty directly in the derivative of a CLF, thus allowing a low dimensional representation of the uncertainty. In our second contribution, we demonstrate the practicality of PSS by incorporating it into an episodic learning algorithm.

Our paper is organized as follows. Section \ref{sec:prelims} reviews CLFs and ISS. Section \ref{sec:pss} defines Projection to State Stability (PSS), and how PSS enables constructing bounds on the state of a system that depend on a projected disturbance. Section \ref{sec:modeling} defines a broad class of model uncertainty for affine control systems, evaluates how this uncertainty impacts the Lyapunov derivative, and demonstrates how to restrict this uncertainty with data to determine if a system is PSS. Section \ref{sec:learning} discusses how episodic learning can be used to improve PSS guarantees in practice, and presents simulation results with an uncertain inverted pendulum model.


\section{PRELIMINARIES}
\label{sec:prelims}
This section provides a review of Control Lyapunov Functions (CLFs) and Input to State Stability (ISS). These tools will be used in Section \ref{sec:pss} to define Projection to State Stability. 
This section concludes with a brief discussion of how these definitions must be modified to hold over a restriction of the domain.

Consider a state space $\cal{X} \subseteq \R^n$ and a control input space $\cal{U} \subseteq \R^m$. Assume that $\cal{X}$ is path-connected and that $\mb{0} \in \cal{X}$. Consider a system governed by:
\begin{equation}\label{eq:controlsys}
    \dot{\mb{x}} = \mb{f}(\mb{x}, \mb{u}),
\end{equation}
for state $\mb{x} \in \cal{X}$ and its derivative $\dot{\mb{x}}$, control input $\mb{u} \in \cal{U}$, and dynamics $\mb{f}: \cal{X} \times \cal{U} \to \R^n$. In this paper we assume $\mb{f}$ is locally Lipschitz continuous. The following definitions, taken from \cite{Khalil}, are useful in analyzing stability of \eqref{eq:controlsys}.

\begin{definition}[\it{Class $\cal{K}$ Function}]
A continuous function $\alpha: [0, a) \to \R_+$, with $a > 0$, is class $\cal{K}$, denoted $\alpha \in \cal{K}$, if it is monotonically (strictly) increasing and satisfies $\alpha(0) = 0$. If the domain of $\alpha$ is all of $\R_+$ and $\lim_{r \to \infty} \alpha(r) = \infty$, then $\alpha$ is termed radially unbounded and class $\cal{K}_{\infty}$.
\end{definition}

\begin{definition}[\it{Class $\cal{KL}$ Function}]
A continuous function $\beta: [0, a) \times \R_+ \to \R_+$, with $a > 0$, is class $\cal{KL}$, denoted $\beta \in \cal{KL}$, if the function $r \mapsto \beta(r, s) \in \cal{K}$ for all $s \in \R_+$, and the function $s \mapsto \beta(r, s)$ is monotonically non-increasing with $\beta(r,s)\rightarrow 0$ as $s \rightarrow \infty$ for all $r \in [0, a)$.
\end{definition}
We note that the strictly increasing nature of Class $\cal{K}$ $(\cal{K}_\infty)$ functions permits an inverse Class $\cal{K}$ $(\cal{K}_\infty)$ function $\alpha^{-1}:[0,\alpha(a))\to \R_+$. We also note that the composition of Class $\cal{K}$ ($\cal{K}_\infty$) functions is itself a Class $\cal{K}$ ($\cal{K}_\infty$) function. Given these definitions, we define Control Lyapunov Functions (CLFs) as in \cite{artstein1983stabilization}, \cite{lin1991universal}.

\begin{definition}[\it{Control Lyapunov Function}]\label{def:clf}
A continuously differentiable function $V: \cal{X} \to \R_+$ is a CLF for $(\ref{eq:controlsys})$ on $\cal{X}$ if there exist $\underline{\alpha}, \overline{\alpha}, \alpha \in \cal{K}_{\infty}$ such that:
\begin{align}\label{eq:clf}
    \underline{\alpha}(\norm{ \mb{x} }) \leq V(\mb{x}) &\leq \overline{\alpha}(\norm{\mb{x}})\nonumber\\
    \inf_{\mb{u} \in \cal{U}} \dot{V}(\mb{x}, \mb{u}) &\leq -\alpha(\norm{\mb{x}}),
\end{align}
for all $\mb{x} \in \cal{X}$.
\end{definition}
If there exists a CLF for a system, then a state-feedback controller $\mb{k}: \cal{X} \to \cal{U}$ can be selected such that $\mb{0}$ is a globally asymptotically stable equilibrium point. In particular, for all $\mb{x} \in \cal{X}$, $\mb{k}(\mb{x})$ should be chosen such that $\dot{V}(\mb{x}, \mb{k}(\mb{x})) \leq -\alpha(\norm{\mb{x}})$. We note that $\underline{\alpha}, \overline{\alpha}, \alpha$ only need to be Class $\cal{K}$ for this definition, but we extend them to $\cal{K}_\infty$ to simplify later analysis.

To accommodate disturbances or uncertainties, we consider a disturbance space $\cal{D} \subseteq \R^d$, and a modified system:
\begin{equation}\label{eq:distsys}
    \dot{\mb{x}} = \mb{f}(\mb{x}, \mb{u}, \mb{d}),
\end{equation}
for disturbance $\mb{d} \in \cal{D}$ and dynamics $\mb{f}: \cal{X} \times \cal{U} \times \cal{D} \to \R^n$. We again assume $\mb{f}$ is locally Lipschitz continuous. The disturbance may be time-varying, state-dependent, and/or input-dependent. We assume that the disturbance is bounded for almost all times $t \geq 0$ (essentially bounded in time). This leads to the definition of ISS and ISS-CLFs as formulated in \cite{sontag1989smooth}, \cite{sontag1995characterizations}.

\begin{definition}[\it{Input to State Stability}]
    \label{def:iss}
    Given a state-feedback controller $\mb{k}:\cal{X}\to\cal{U}$, a system is Input to State Stable (ISS) if there exist $\beta \in \cal{KL}_{\infty}$ and $\gamma \in \cal{K}_{\infty}$ such that the solution to (\ref{eq:distsys}) satisfies:
    \begin{equation} \label{eqn:iss}
        \norm{\mb{x}(t)} \leq \beta(\norm{\mb{x}(0)}, t) + \gamma\left( \sup_{\tau \geq 0} \norm{\mb{d}(\tau)} \right),
    \end{equation}
    for all $t \geq 0$.
\end{definition}

\begin{definition}[\it{Input to State Stable Control Lyapunov Function}]\label{def:issclf}
A continuously differentiable function $V: \cal{X} \to \R_+$ is an Input to State Stable Control Lyapunov Function (ISS-CLF) for (\ref{eq:distsys}) on $\cal{X}$ if there exist $\underline{\alpha}, \overline{\alpha}, \alpha, \rho \in \cal{K}_{\infty}$ such that:
\begin{align}\label{eq:issclf}
    \underline{\alpha}(\norm{ \mb{x} }) \leq V(\mb{x}) &\leq \overline{\alpha}(\norm{\mb{x}})\nonumber\\
    \norm{\mb{x}} \geq \rho(\norm{\mb{d}}) \implies \inf_{\mb{u} \in \cal{U}} \dot{V}(\mb{x}, \mb{u}, \mb{d}) &\leq -\alpha(\norm{\mb{x}}),
\end{align}
for all $\mb{x} \in \cal{X}$ and $\mb{d} \in \cal{D}$. 
\end{definition}
As with CLFs, if there exists an ISS-CLF for a system, then a state-feedback controller $\mb{k}: \cal{X} \to \cal{U}$ can be chosen such that the system is ISS. If the disturbance is input-dependent, it is additionally required that $\mb{k}$ induces essentially bounded disturbances in time.

The condition on the Lyapunov function derivative in (\ref{eq:clf}) or (\ref{eq:issclf}) may not be satisfied on the entire state space $\cal{X}$. In particular it may only be satisfied on a subset $\cal{C} \subseteq \cal{X}$. The system may leave $\cal{C}$ during its evolution, implying the desired derivative condition may no longer be satisfiable. We therefore consider the following definition and lemma.
\begin{definition}[\it{Forward Invariance}]
    Consider the system governed by (\ref{eq:controlsys}). A subset $\cal{F} \subseteq \cal{X}$ is forward invariant if there exists a state-feedback controller $\mb{k}: \cal{X} \to \cal{U}$ such that $\mb{x}(0) \in \cal{F}$ implies $\mb{x}(t) \in \cal{F}$ for all $t \geq 0$.
\end{definition}

The definition of forward invariance applies to systems governed by (\ref{eq:distsys}), with disturbances appropriately restricted to subsets of $\cal{D}$ if the disturbances are modeled as state-dependent and/or input-dependent. If $\mb{0}\in\cal{C}$, we may restrict Definitions \ref{def:clf} and \ref{def:issclf} to a forward invariant subset $\cal{F} \subseteq \cal{C}$ with $\mb{0} \in \cal{F}$, provided such a subset exists.

\begin{lemma}\label{lem:issomega}
    A sublevel set $\Omega \subseteq \cal{X}$ of an ISS-CLF $V$ is a forward invariant set, provided $\norm{\mb{x}} \geq \rho(\norm{\mb{d}})$ for all $\mb{x} \in \partial \Omega$ and appropriately restricted $\mb{d} \in \cal{D}$.
\end{lemma}
\begin{proof}
 The condition on the Lyapunov derivative in (\ref{eq:issclf}) implies the existence of a state-feedback controller $\mb{k}: \cal{X} \to \cal{U}$ satisfying $\dot{V}(\mb{x}, \mb{k}(\mb{x}), \mb{d}) \leq -\alpha(\norm{\mb{x}})$ for all $\mb{x} \in \partial \Omega$ and appropriately restricted $\mb{d} \in \cal{D}$. Let $c = V(\mb{x})$ for any $\mb{x} \in \partial \Omega$. If $V(\mb{x}(0)) \in [0, c]$, then $V(\mb{x}(t)) \in [0, c]$ for all $t > 0$ by Nagumo's Theorem \cite{nagumo1942lage}, \cite{abraham2012manifolds}. Thus, if $\mb{x}(0) \in \Omega$, then $\mb{x}(t) \in \Omega$ for all $t \geq 0$.
\end{proof}

\section{Projection to State Stability}
\label{sec:pss}
Input to State Stability (ISS) requires a bound on the state in terms of the norm of the disturbance as it appears in the state dynamics (see Definition \ref{def:iss} in Section \ref{sec:prelims}). This requirement does not easily permit analysis of Input to State behavior when the disturbance is more easily described by its impact in a Lyapunov function derivative. This limitation motivates Projection to State Stability (PSS), which instead relies a bound on the state in terms of a projection of the disturbance.

    

\begin{definition}[\it{Dynamic Projection}]
    A continuously differentiable function $\bs{\Pi}: \cal{X}\to\R^k$ is a dynamic projection if there exist $\underline{\sigma}, \overline{\sigma} \in \cal{K}_{\infty}$ satisfying: \begin{equation}\label{eqn:picond}
        \underline{\sigma}(\norm{\mb{x}}) \leq \norm{\bs{\Pi}(\mb{x})} \leq \overline{\sigma}(\norm{\mb{x}}),
    \end{equation}
    for all $\mb{x} \in \cal{X}$.    
\end{definition}

Let $\cal{Y} = \mathrm{range}(\bs{\Pi})$, and let $\mb{y} = \bs{\Pi}(\mb{x})$ for all $\mb{x} \in \cal{X}$. Consider the system governed by (\ref{eq:distsys}). The associated projected system is governed by the dynamics:
\begin{equation}\label{eqn:projdyn}
    \dot{\mb{y}} = \mb{D}_{\bs{\Pi}}(\mb{x})\mb{f}(\mb{x}, \mb{u}, \mb{0}) + \underbrace{\mb{D}_{\bs{\Pi}}(\mb{x})(\mb{f}(\mb{x}, \mb{u}, \mb{d}) - \mb{f}(\mb{x}, \mb{u}, \mb{0}))}_{\bs{\delta}},
\end{equation}
where $\mb{D}_{\bs{\Pi}}: \cal{X} \to \R^{k \times n}$ denotes the Jacobian of $\bs{\Pi}$, and $\bs{\delta}$ is implicitly a function of $\mb{x}$, $\mb{u}$, and $\mb{d}$. For the following definitions, we assume $\bs{\delta}$ is essentially bounded in time. 

We are now ready to state our main definition. 
The key difference between PSS and ISS (Definition \ref{def:iss}) is the use of $\bs{\delta}$ \eqref{eqn:projdyn} rather than the native disturbance $\mb{d}$.

    
\begin{definition}[\it{Projection to State Stability}]
    Given a state-feedback controller $\mb{k}:\cal{X}\to\cal{U}$, a system is Projection to State Stable (PSS) with respect to the projection $\bs{\Pi}$ if there exist $\beta \in \cal{KL}_{\infty}$ and $\gamma \in \cal{K}_{\infty}$ such that the solution to (\ref{eq:distsys}) satisfies:
   \begin{equation}
        \norm{\mb{x}(t)} \leq \beta(\norm{\mb{x}(0)}, t) + \gamma\left( \sup_{\tau \geq 0}  \norm{\,\bs{\delta}(\tau)\,} \right),
    \end{equation}
    for all $t \geq 0$, with $\bs{\delta}$ as defined in (\ref{eqn:projdyn}).
\end{definition}

\begin{remark}
    If $\bs{\Pi}$ is an inclusion map with $k = n$, and the system can be specified as:
    \begin{equation}\label{eqn:affdistrep}
        \mb{f}(\mb{x}, \mb{u}, \mb{d}) = \mb{f}(\mb{x}, \mb{u}, \mb{0}) + \mb{d},
    \end{equation}
 then PSS is equivalent to ISS.
\end{remark}
Similarly, we can also construct a Lyapunov function  that certifies a system is PSS with respect to a projection.  
\begin{definition}[\it{Projection to State Stable Control Lyapunov Function}]
\label{def:SS-CLF}
   A continuously differentiable function $W: \cal{Y} \to \R_+$ is a Projection to State Stable Control Lyapunov Function (PSS-CLF) for (\ref{eqn:projdyn}) on $\cal{X}$ if there exist $\underline{\alpha}, \overline{\alpha}, \alpha, \rho \in \cal{K}_{\infty}$ satisfying:
    \begin{align}\label{eqn:newconds}
        \underline{\alpha}(\norm{\bs{\Pi}(\mb{x})}) \leq W(\bs{\Pi}(\mb{x})) &\leq \overline{\alpha}(\norm{\bs{\Pi}(\mb{x})}) \nonumber\\
        \norm{\bs{\Pi}(\mb{x})} \geq \rho(\norm{\bs{\delta}}) \implies \inf_{\mb{u} \in \cal{U}} \dot{W}(\mb{x}, \mb{u}, \bs{\delta}) &\leq -\alpha(\norm{\bs{\Pi}(\mb{x})}),
    \end{align}
    for all $\mb{x} \in \cal{X}$.
\end{definition}
As with ISS-CLFs, this definition can be restricted to a forward invariant set containing $\mb{0}$. We now show that a PSS-CLF certifies a system is PSS.

    
\begin{theorem}\label{thm:equivalence}
    If the system governed by (\ref{eqn:projdyn}) has a PSS-CLF, then the system governed by (\ref{eq:distsys}) is PSS with respect to the projection $\bs{\Pi}$.
\end{theorem}
    
\begin{proof}
    The bounds in (\ref{eqn:newconds}) can be weakened to:
    \begin{align}\label{eqn:newnewconds}
        &W(\bs{\Pi}(\mb{x})) \geq \overline{\alpha}\circ\rho(\norm{\bs{\delta}}) \nonumber\\
        & \qquad \implies \inf_{\mb{u} \in \cal{U}} \dot{W}(\mb{x}, \mb{u}, \bs{\delta}) \leq -\alpha \circ \overline{\alpha}^{-1}(W (\bs{\Pi}(\mb{x}))).
    \end{align}
    That is, if (\ref{eqn:newnewconds}) holds, (\ref{eqn:newconds}) holds. Therefore, a choice of state-feedback controller exists such that the system governed by (\ref{eqn:projdyn}) is Input to State Stable (ISS) with $\bs{\delta}$ viewed a disturbance. This implies that there exist $\beta \in \cal{KL}_{\infty}$ and $\gamma \in \cal{K}_{\infty}$ such that:
    \begin{equation}
        \norm{\bs{\Pi}(\mb{x}(t))} \leq \beta(\norm{\bs{\Pi}(\mb{x}(0))}, t) + \gamma\left( \sup_{\tau \geq 0} \norm{\bs{\delta}(\tau)} \right),
    \end{equation}
    for all $t \geq 0$. Since $\bs{\Pi}$ satisfies (\ref{eqn:picond}) we have:
    \begin{align}
        &\norm{\mb{x}(t)} \leq \underline{\sigma}^{-1}\left(\beta(\overline{\sigma}(\norm{\mb{x}(0)}), t) + \gamma\left( \sup_{\tau \geq 0} \norm{\bs{\delta}(\tau)} \right) \right).
    \end{align}
    Finally, define $\beta' \in \cal{KL}_{\infty}$ and $\gamma' \in \cal{K}_{\infty}$ as:
    \begin{align}
        \beta'(r, s) &= \underline{\sigma}^{-1}(2\beta(\overline{\sigma}(r), s))\\
        \gamma'(r) &= \underline{\sigma}^{-1}(2\gamma(r)).
    \end{align}
    From the weak form of the triangle inequality presented in \cite{sontag1989smooth}, \cite{kellett2014compendium}, it follows that:
    \begin{equation}
        \norm{\mb{x}(t)} \leq \beta'(\norm{\mb{x}(0)}, t) + \gamma'\left( \sup_{\tau \geq 0} \norm{\bs{\delta}(\tau)} \right).
    \end{equation}
\end{proof}

We next show that a CLF $V$ for the undisturbed dynamics of a system can be viewed as a projection, thus yielding a PSS-CLF that certifies PSS with respect to $V$.
    
\begin{corollary}\label{cor:clftopss}
    Suppose $V: \cal{X} \to \R_+$ is a CLF on $\cal{X}$ for the system
    $ \dot{\mb{x}} = \mb{f}(\mb{x}, \mb{u}, \mb{0}).$
    Then the disturbed system governed by (\ref{eq:distsys}) is PSS with respect to the projection $V$.
\end{corollary}
    
\begin{proof}
    With the projection $V$ we have that:
    \begin{equation}
        \delta = \grad V(\mb{x})^\top(\mb{f}(\mb{x}, \mb{u}, \mb{d}) - \mb{f}(\mb{x}, \mb{u}, \mb{0})).
    \end{equation}
    where  $\grad V: \cal{X} \to \R^n$ is the gradient of the Lyapunov function.
    The projected system is governed by:
    \begin{equation}\label{eqn:deltaderiv}
        \dot{V}(\mb{x},\mb{u},\delta) =  \grad V(\mb{x})^\top\mb{f}(\mb{x}, \mb{u}, \mb{0}) + \delta,    
    \end{equation}
    Since $V$ is a CLF,
    there exists a state-feedback controller $\mb{k}: \cal{X} \to \cal{U}$ satisfying:
    \begin{equation}
        \dot{V}(\mb{x}, \mb{k}(\mb{x}), 0) \leq -\alpha(\norm{\mb{x}}),
    \end{equation}
    for all $\mb{x} \in \cal{X}$. Let $\alpha_p, \alpha_q \in \cal{K}_{\infty}$ satisfy $\alpha_p + \alpha_q = \mb{\alpha}$. Then:
    \begin{align}
        \dot{V}(\mb{x}, \mb{k}(\mb{x}), \delta) &\leq -\alpha(\norm{\mb{x}}) + \delta  \nonumber\\
        &\leq -\alpha_p(\norm{\mb{x}}) - \alpha_q(\norm{\mb{x}}) + \vert \delta \vert.
    \end{align}
    Therefore:
    \begin{equation}
        \norm{\mb{x}} \geq \alpha_q^{-1}(\vert \delta \vert) \implies \dot{V}(\mb{x}, \mb{k}(\mb{x}), \delta) \leq -\alpha_p(\norm{\mb{x}}).
    \end{equation}
    Since $V$ is a CLF we may weaken the bounds as in the proof of Theorem \ref{thm:equivalence} to:
    \begin{align}
        &V(\mb{x}) \geq \overline{\alpha} \circ \alpha_q^{-1}(\vert\delta\vert) \nonumber \\
        & \qquad \implies \dot{V}(\mb{x}, \mb{k}(\mb{x}), \delta) \leq -\alpha_p \circ \overline{\alpha}^{-1}(V(\mb{x})),
    \end{align}
    noting that $\overline{\alpha} \circ \alpha_q^{-1}$ and $\alpha_p \circ \overline{\alpha}^{-1}$ are class $\cal{K}_{\infty}$. It follows from Definition \ref{def:SS-CLF} that the identity map on $\R_+$ is a PSS-CLF for (\ref{eqn:deltaderiv}). Therefore, the system (\ref{eq:distsys}) is PSS with respect to the projection $V$ by Theorem \ref{thm:equivalence}.
\end{proof}

\section{UNCERTAINTY MODELING \& ANALYSIS}
\label{sec:modeling}
In this section we consider a structured form of uncertainty present in affine control systems. We analyze the impact of this uncertainty on a Lyapunov function derivative, and on the PSS behavior of the system.

\subsection{Uncertain Affine Systems}

We consider affine control systems of the form:
\begin{equation}\label{eq:truesys}
    \dot{\mb{x}} = \mb{f}(\mb{x}) + \mb{g}(\mb{x})\mb{u},
\end{equation}
with drift dynamics $\mb{f}: \cal{X} \to \R^n$ and actuation matrix $\mb{g}: \cal{X} \to \R^{n \times m}$. If $\mb{f}$ and $\mb{g}$ are unknown, we may consider an estimated model of the system:
\begin{equation}\label{eq:nomsys}
    \dot{\mb{x}} = \hat{\mb{f}}(\mb{x}) + \hat{\mb{g}}(\mb{x})\mb{u},
\end{equation}
where $\hat{\mb{f}}: \cal{X} \to \R^n$ and  $\hat{\mb{g}}: \cal{X} \to \R^{n \times m}$ are estimates of $\mb{f}$ and $\mb{g}$, respectively. In this case, (\ref{eq:truesys}) can be expressed as:
\begin{equation}\label{eq:distaffsys}
    \dot{\mb{x}} = \hat{\mb{f}}(\mb{x}) + \hat{\mb{g}}(\mb{x})\mb{u} + \overbrace{ (\underbrace{\mb{g}(\mb{x}) - \hat{\mb{g}}(\mb{x})}_{\mb{A}(\mb{x})})\mb{u} + \underbrace{\mb{f}(\mb{x}) - \hat{\mb{f}}(\mb{x})}_{\mb{b}(\mb{x})}}^{\mb{d}},
\end{equation}
obtaining a representation of the dynamics as in (\ref{eqn:affdistrep}). Note that the disturbance $\mb{d} = \mb{A}(\mb{x})\mb{u} + \mb{b}(\mb{x})$ is explicitly characterized as time-invariant, state-dependent, and input-dependent, with potentially unknown $\mb{A}(\mb{x})$ and $\mb{b}(\mb{x})$ for all $\mb{x} \in \cal{X}$.

As discussed in \cite{ames2014rapidly}, \cite{taylor2019episodic}, CLFs may be constructively formed for affine systems under proper assumptions regarding relative degree and unbounded control. Furthermore, if the true system satisfies the relative degree properties of the estimated model, then the CLF found for the estimated system can be used for the true system. 

Assume $\mb{f}$, $\mb{g}$, $\hat{\mb{f}}$, and $\hat{\mb{g}}$ are Lipschitz continuous (implying $\mb{A}$ and $\mb{b}$ are Lipschitz continuous), and let $V$ be a CLF candidate for \eqref{eq:nomsys}. The time derivative of $V$ is given by:
\begin{align}\label{eq:lyapder}
    \dot{V}(\mb{x}, \mb{u}, \mb{d}) =& \overbrace{( \hat{\mb{f}}(\mb{x}) + \hat{\mb{g}}(\mb{x})\mb{u} )^\top \grad{V}(\mb{x})}^{\hat{\dot{V}}(\mb{x}, \mb{u})} \nonumber\\
    & + (\underbrace{\mb{A}(\mb{x})^\top\grad V(\mb{x})}_{\mb{a}(\mb{x})})^\top\mb{u} + \underbrace{\mb{b}(\mb{x})^\top\grad V(\mb{x})}_{b(\mb{x})},
\end{align}
for all $\mb{x} \in \cal{X}$ and $\mb{u} \in \cal{U}$. As proposed in \cite{taylor2019episodic}, we may wish to reduce the estimation error $\vert \dot{V}-\hat{\dot{V}} \vert$ by improving $\hat{\dot{V}}$ with estimates of $\mb{a}$ and $b$. Given continuous estimators $\hat{\mb{a}}: \cal{X} \to \R^m$ and $\hat{b}: \cal{X} \to \R$, \eqref{eq:lyapder} may be reformulated as:
\begin{align}\label{eq:lyapderest}
    & \dot{V}(\mb{x}, \mb{u}, \mb{d}) = \overbrace{( \hat{\mb{f}}(\mb{x}) + \hat{\mb{g}}(\mb{x})\mb{u} )^\top \grad{V}(\mb{x}) + \hat{\mb{a}}(\mb{x})^\top\mb{u} + \hat{b}(\mb{x})}^{\hat{\dot{V}}(\mb{x}, \mb{u})} \nonumber\\
    & \quad + (\underbrace{\mb{A}(\mb{x})^\top\grad V(\mb{x}) - \hat{\mb{a}}(\mb{x})}_{\mb{a}(\mb{x})})^\top\mb{u} + \underbrace{\mb{b}(\mb{x})^\top\grad V(\mb{x}) - \hat{b}(\mb{x})}_{b(\mb{x})},
\end{align}
for all $\mb{x} \in \cal{X}$ and $\mb{u} \in \cal{U}$.

Both formulations decompose $\dot{V}$ into an estimated component, $\hat{\dot{V}}$, and a residual component. In (\ref{eq:lyapder}) the residual terms $\mb{a}$ and $b$ capture the effect of the unmodeled dynamics on the Lyapunov function derivative. In \eqref{eq:lyapderest} the residual terms reflect the error in estimating this effect. Additionally, viewing $V$ as a projection results in $\delta = \mb{a}(\mb{x})^\top\mb{u}+b(\mb{x})$. 

\subsection{Projection  to State Stability via Uncertainty Functions}
If knowledge on what values $\mb{a}$ and $b$ can assume is available, the impact on the Lyapunov derivative can be constrained in a manner permitting PSS analysis of a system. Therefore, we define a function characterizing the possible uncertainties at a given state.

\begin{definition}[\it{Uncertainty Function}]
    Let $\cal{P}(\R^m \times \R)$ denote the set of all subsets of $\R^m \times \R$. An uncertainty function for (\ref{eq:lyapder}) or (\ref{eq:lyapderest}) is a function $\Delta: \cal{X} \to \cal{P}(\R^m \times \R)$ with $\Delta(\mb{x})$ bounded and satisfying $(\mb{a}(\mb{x}), b(\mb{x})) \in \Delta(\mb{x})$ for all $\mb{x} \in \cal{X}$.
\end{definition}

For a given $\mb{x} \in \cal{X}$, we refer to $\Delta(\mb{x})$ as an uncertainty set. Suppose there exists a valid uncertainty function $\Delta$ for (\ref{eq:lyapder}) or  \eqref{eq:lyapderest}. Then $V$ satisfies:
\begin{equation}\label{eq:lyapbound}
    \dot{V}(\mb{x}, \mb{u}, \delta) \leq \hat{\dot{V}}(\mb{x}, \mb{u}) + \sup_{(\mb{a}, b) \in \Delta(\mb{x})} (\mb{a}^\top \mb{u} + b),
\end{equation}
for all $\mb{x} \in \cal{X}$ and $\mb{u} \in \cal{U}$. One major challenge is to define a $\Delta$ that is non-vacuous and thus practically relevant.  
From this point forward we limit our attention to a subset of the state space and make a critical assumption regarding the estimate $\hat{\dot{V}}$ for a CLF $V$.
\begin{assumption}\label{as:clfdercond}
    Let $V$ be a CLF for the system governed by (\ref{eq:nomsys}) on a subset $\cal{C} \subseteq \cal{X}$ with $\mb{0} \in \cal{C}$. We assume that:
    \begin{equation}\label{eqn:asscond}
        \inf_{\mb{u} \in \cal{U}} \hat{\dot{V}}(\mb{x}, \mb{u}) \leq -\alpha(\norm{\mb{x}}).
    \end{equation}
    for all $\mb{x} \in \cal{C}$. If $\hat{\dot{V}}$ is specified as in (\ref{eq:lyapder}), then this assumption is satisfied by definition. If $\hat{\dot{V}}$ is specified as in (\ref{eq:lyapderest}), then this assumption states that the addition of the estimators $\hat{\mb{a}}$ and $\hat{b}$ does not make it impossible to choose a control input such that (\ref{eqn:asscond}) is satisfied.
\end{assumption}
If the estimated and true system satisfy the same relative degree property, then this assumption amounts to the addition of estimates $\hat{\mb{a}}$ and $\hat{b}$ not violating the relative degree property.

\begin{assumption}
    Let $\mb{A}$ and $\mb{b}$ be defined as in (\ref{eq:distaffsys}), and let $\cal{C}$ be defined as in Assumption \ref{as:clfdercond}. We assume $\mb{A}$ and $\mb{b}$ are bounded on $\cal{C}$.
\end{assumption}
If $\cal{C}$ is compact, this assumption is automatically satisfied since $\mb{A}$ and $\mb{b}$ are assumed continuous. Under Assumption \ref{as:clfdercond}, the set of admissible control inputs $\cal{U}(\mb{x})$:
\begin{equation}\label{eq:adcontrol}
    \cal{U}(\mb{x}) = \{ \mb{u} \in \cal{U} : \hat{\dot{V}}(\mb{x}, \mb{u}) \leq -\alpha(\norm{ \mb{x} }) \},
\end{equation}
is non-empty, for all $\mb{x} \in \cal{C}$. Then the CLF $V$  satisfies:
\begin{align}
    \underline{\alpha}(\norm{ \mb{x} }) \leq V(\mb{x}) &\leq \overline{\alpha}(\norm{\mb{x}})\nonumber\\
    \inf_{\mb{u} \in \cal{U}(\mb{x})} \dot{V}(\mb{x}, \mb{u}, \delta) - \sup_{(\mb{a}, b) \in \Delta(\mb{x})} (\mb{a}^\top \mb{u} + b) &\leq -\alpha(\norm{\mb{x}}),
\end{align}
for all $\mb{x} \in \cal{C}$. We now develop sufficient conditions on the uncertainty function that certifies  \eqref{eq:distaffsys} as PSS with respect to the CLF $V$ (with $V$ interpreted as a projection).

\begin{theorem}[\it{Sufficient Conditions for PSS in Affine Control Systems}]\label{thm:isscond}
    Consider the system in (\ref{eq:distaffsys}), and a CLF $V$ for (\ref{eq:nomsys}) with estimated time-derivative $\hat{\dot{V}}$ as defined in (\ref{eq:lyapder}) or (\ref{eq:lyapderest}), satisfying Assumption \ref{as:clfdercond}. Let $\Delta$ be an uncertainty function and let $\mb{k}: \cal{X} \to \cal{U}$ be a state-feedback controller satisfying $\mb{k}(\mb{x}) \in \cal{U}(\mb{x})$ for all $\mb{x} \in \cal{C}$, with $\cal{U}(\mb{x})$ defined as in (\ref{eq:adcontrol}). Suppose there exists $\alpha_p, \alpha_q \in \cal{K}_{\infty}$ with $\alpha_p + \alpha_q = \alpha$ and a sublevel set $\Omega \subseteq \cal{C}$ of $V$ satisfying:
    \begin{equation}\label{eq:boundarycond}
        \norm{\mb{x}} \geq \sup_{(\mb{a}, b) \in \Delta(\mb{x})} \alpha_q^{-1}(\mb{a}^\top\mb{k}(\mb{x}) + b),
    \end{equation}
    for all $\mb{x} \in \partial\Omega$. Then the system governed by (\ref{eq:distaffsys}) is PSS with respect to the projection $V$ on $\Omega$.
\end{theorem}

\begin{proof}
    First, note that:
    \begin{align}
        &\dot{V}(\mb{x}, \mb{k}(\mb{x}), \delta) - \sup_{(\mb{a}, b) \in \Delta(\mb{x})} (\mb{a}^\top \mb{k}(\mb{x}) + b) \leq -\alpha(\norm{\mb{x}}) \nonumber\\ 
        &\qquad\qquad\qquad\qquad\qquad\quad = -\alpha_p(\norm{\mb{x}}) - \alpha_q(\norm{\mb{x}}),
    \end{align}
    for all $\mb{x} \in \cal{C}$. Since (\ref{eq:boundarycond}) holds for all $\mb{x} \in \partial \Omega$ and $\alpha_q$ is monotonically increasing, we have:
    \begin{equation}
        \alpha_q(\norm{\mb{x}}) \geq \sup_{(\mb{a}, b) \in \Delta(\mb{x})} (\mb{a}^\top\mb{k}(\mb{x}) + b),
    \end{equation}
    for all $\mb{x} \in \partial \Omega$. It follows that:
    \begin{equation}
        \dot{V}(\mb{x}, \mb{k}(\mb{x}), \delta) \leq -\alpha_p(\norm{\mb{x}}),
    \end{equation}
    for all $\mb{x} \in \partial \Omega$. This means $\Omega$ is forward invariant, with a proof similar to that of Lemma \ref{lem:issomega}. Since $V$ is a CLF for (\ref{eq:nomsys}), Corollary \ref{cor:clftopss} can be restricted to $\Omega$; that is, the system is PSS with respect to the projection $V$ on $\Omega$. 
\end{proof}

    
    
    
    

    
We may want to study a particular set of interest $\cal{E}$ over which the impact of the uncertainty can be bounded. For $r > 0$, let $B_r$ be the open ball around $\mb{0}$ of radius $r$, typically used to define a ball contained in $\cal{E}$ in the subsequent analysis.
\begin{corollary}\label{cor:boundeddist}
    Suppose there is a set $\cal{E}$ and $\mu \geq 0$ satisfying:
    \begin{equation}
        \sup_{(\mb{a}, b) \in \Delta(\mb{x})} (\mb{a}^\top\mb{k}(\mb{x}) + b) \leq \mu,
    \end{equation}
    for all $\mb{x} \in \cal{E}$. If there exists a sublevel set $\Omega$ of $V$ such that:
    \begin{equation}
        B_{\alpha_q^{-1}(\mu)} \subseteq \Omega \subseteq \cal{C} \intersect \cal{E},
    \end{equation}
    then the system is PSS with respect to the (CLF) projection $V$ on $\Omega$, and the smallest sublevel set of $V$ containing $B_{\alpha_q^{-1}(\mu)}$ is asymptotically stable.
\end{corollary}

\begin{proof}
    First, note that:
    \begin{equation}
        \norm{\mb{x}} \geq \alpha_q^{-1}(\mu) \geq \sup_{(\mb{a}, b) \in \Delta(\mb{x})} \alpha_q^{-1}(\mb{a}^\top\mb{k}(\mb{x}) + b),
    \end{equation}
    for all $\mb{x} \in \partial \Omega$, and the system is PSS on $\Omega$ by Theorem \ref{thm:isscond}. The smallest sublevel set of $V$ containing $B_{\alpha_q^{-1}(\mu)}$ is asymptotically stable since:
    \begin{equation}
        \norm{\mb{x}} \geq \alpha_q^{-1}(\mu) \geq \implies \dot{V}(\mb{x}, \mb{k}(\mb{x}), \delta) \leq -\alpha_p(\norm{\mb{x}}).
    \end{equation}
\end{proof}

Improving the uncertainty set (e.g., reducing uncertainty using learning) directly leads to larger sets for a given bound, or tighter bounds on a given set.  We state this formally in the next result.

\begin{corollary}[\it{Uncertainty Function Improvement}]
    \label{cor:usimprove}
    Consider uncertainty functions $\Delta$ and $\Delta'$, as well as $\cal{E}$ and $\mu$ as defined in Corollary \ref{cor:boundeddist}.
    \begin{itemize}
        \item Fix $\mu > 0$ and let $\cal{E}_{\mu}$ be defined as:
        \begin{equation}
            \cal{E}_{\mu} = \{ \mb{x} \in \cal{X} : \sup_{(\mb{a}, b) \in \Delta(\mb{x})} (\mb{a}^\top\mb{k}(\mb{x}) + b) \leq \mu \}.
        \end{equation} 
        \item Fix $\cal{E} \subseteq \cal{X}$ and let $\mu_{\cal{E}}$ be defined as:
        \begin{equation}
            \mu_{\cal{E}} = \sup_{\vphantom{(\mb{a}, b) \in \Delta(\mb{x})}\mb{x} \in \cal{E}} \sup_{(\mb{a}, b) \in \Delta(\mb{x})} (\mb{a}^\top\mb{k}(\mb{x}) + b).
        \end{equation}
    \end{itemize}
    Suppose $\Delta'(\mb{x}) \subseteq \Delta(\mb{x})$ for all $\mb{x} \in \cal{X}$. Then the associated set $\cal{E}_{\mu}'$ and scalar $\mu_{\cal{E}}'$ satisfy $\cal{E}_{\mu} \subseteq \cal{E}_{\mu}'$ and $\mu_{\cal{E}}' \leq \mu_{\cal{E}}$.
\end{corollary}

\begin{proof}
    \begin{equation}
        \sup_{(\mb{a}, b) \in \Delta'(\mb{x})} (\mb{a}^\top\mb{k}(\mb{x}) + b) \leq \sup_{(\mb{a}, b) \in \Delta(\mb{x})} (\mb{a}^\top\mb{k}(\mb{x}) + b).
        \vspace{-0.1in}
    \end{equation}
\end{proof}
\subsection{Uncertainty Function Construction}
\label{sec:construction}
We now provide a constructive method for creating an uncertainty function from a dataset of of state and control values generated by a system. Assume $\mb{A}$ and $\mb{b}$ are Lipschitz continuous with constants $L_{\mb{A}}$ and $L_{\mb{b}}$, respectively. Additionally, assume that $\mb{A}$ and $\mb{b}$ are bounded on $\cal{C}$ by constants $\norm{\mb{A}}_{\infty}$ and $\norm{\mb{b}}_{\infty}$, respectively. Consider a dataset $D \subseteq (\cal{X} \times \cal{U}) \times \R$ consisting of data-measurement pairs $((\mb{x}, \mb{u}), \dot{V}(\mb{x}, \mb{u}, \delta))$. Such measurements of $\dot{V}$ can be obtained through numerical differentiation of computed values of $V$. For notational convenience, let $D_0 = \{ (\mb{x}, \mb{u}): ((\mb{x}, \mb{u}), \dot{V}(\mb{x}, \mb{u}, \delta)) \in D \}$. 

\begin{proposition}\label{prop:cons}
    Given a dataset $D$, an uncertainty function $\Delta$ can be constructed as:
    \begin{align}\label{eqn:uncertaintyset}
        \Delta(\mb{x}) =&~ \{ (\mb{a}, b) \in \R^m \times \R: \pm(\mb{a}^\top\mb{u}' + b) \leq \epsilon(\mb{x}, \mb{x}', \mb{u}')  \nonumber\\ 
        & \qquad \textrm{for all } (\mb{x}', \mb{u}') \in D_0 \},
    \end{align}
    for all $\mb{x} \in \cal{X}$, where $\epsilon: \cal{X} \times \cal{X} \times \cal{U} \to \R_+$ is continuous.
\end{proposition}
\begin{remark}
For all $\mb{x}\in\cal{X}$, $\Delta(\mb{x})$ is a closed, symmetric polyhedron and is bounded given sufficiently diverse control inputs in the dataset. In this case, $\Delta(\mb{x})$ is a compact, convex set. The supremum present in Theorem \ref{thm:isscond} and Corollary \ref{cor:boundeddist} becomes a linear program (LP) and can be efficiently solved.
\end{remark}
\begin{proof}[Proof of Proposition \ref{prop:cons}]
    Define observed error as:
    \begin{equation}\label{eqn:obsloss}
        \ell(\mb{x}, \mb{u}) = \left\vert \dot{V}(\mb{x}, \mb{u}, \delta) - \hat{\dot{V}}(\mb{x}, \mb{u}) \right\vert,
    \end{equation}
    for all $(\mb{x}, \mb{u}) \in D_0$. Consider a test point $(\mb{x}, \mb{u}) \in \cal{X} \times \cal{U}$ and a data point $(\mb{x}', \mb{u}') \in D_0$. Note that $\ell(\mb{x}', \mb{u}')$ satisfies:
    \begin{align}
        \ell(\mb{x}', \mb{u}') =&~ \vert \mb{a}(\mb{x}')^\top\mb{u}' + b(\mb{x}') \vert \nonumber\\
        =&~ \vert \mb{a}(\mb{x})^\top \mb{u}' + b(\mb{x}) + (\mb{a}(\mb{x}') - \mb{a}(\mb{x}))^\top\mb{u}' \nonumber \\ & \qquad + b(\mb{x}') - b(\mb{x}) \vert \nonumber\\
        \geq &~ \vert \mb{a}(\mb{x})^\top \mb{u}' + b(\mb{x}) \vert \nonumber \\ & \qquad - \Vert \mb{a}(\mb{x}') - \mb{a}(\mb{x}) \Vert_2 \Vert\mb{u}'\Vert_2 - \vert b(\mb{x}') - b(\mb{x}) \vert,
    \end{align}
    where the inequality follows from the reverse triangle inequality, triangle inequality, and Cauchy-Schwarz inequality. 
    
    For simplicity we proceed with the construction assuming the estimated Lyapunov function derivative is specified as in (\ref{eq:lyapder}). The resulting bound will be modified to include estimators as in (\ref{eq:lyapderest}). Note that:
    \begin{align}
        \Vert \mb{a}(\mb{x}') - \mb{a}(\mb{x}) \Vert_2 = &~ \Vert \mb{A}(\mb{x}')^\top\grad V(\mb{x}') - \mb{A}(\mb{x})^\top\grad V(\mb{x}) \Vert_2 \nonumber\\
        = &~ \Vert (\mb{A}(\mb{x}') - \mb{A}(\mb{x}))^\top\grad V(\mb{x}') \nonumber \\ & \qquad+ \mb{A}(\mb{x})^\top(\grad V(\mb{x}') - \grad V(\mb{x})) \Vert_2 \nonumber\\
        \leq &~ L_{\mb{A}} \Vert \mb{x}' - \mb{x} \Vert_2 \Vert \grad V(\mb{x}') \Vert_2\nonumber \\ &\qquad + \Vert \mb{A} \Vert_{\infty} \Vert \grad V(\mb{x}') - \grad V(\mb{x}) \Vert_2,
    \end{align}
    where the inequality follows from the triangle inequality, submultiplicativity of matrix norms, and Lipschitz and bounded assumptions for $\bf{A}$. Since it is also true that:
    \begin{align}
         \Vert \mb{a}(\mb{x}') - \mb{a}(\mb{x}) \Vert_2 = &~ \Vert \mb{A}(\mb{x}')^\top(\grad V(\mb{x}') - \grad V(\mb{x}))\nonumber \\ &~  + (\mb{A}(\mb{x}') - \mb{A}(\mb{x}))^\top\grad V(\mb{x}) \Vert_2,
    \end{align}
    then the following bound holds:
    \begin{align}
        \Vert \mb{a}(\mb{x}') - \mb{a}(\mb{x}) \Vert_2 \leq &~ L_{\mb{A}} \Vert \mb{x}' - \mb{x} \Vert_2 \Vert \grad V(\mb{x}) \Vert_2 \nonumber \\ &~ + \Vert \mb{A} \Vert_{\infty} \Vert \grad V(\mb{x}') - \grad V(\mb{x}) \Vert_2.
    \end{align}
    Let $\epsilon_{L}(\mb{x}, \mb{x}') = \norm{ \mb{x} - \mb{x}' } \min{\{ \norm{\grad V(\mb{x})}_2, \norm{\grad V(\mb{x}')}_2 \}}$, and let $\epsilon_{\infty}(\mb{x}, \mb{x}') = \norm{\grad V(\mb{x}) - \grad V(\mb{x}')}_2$. Observe that $\epsilon_L$ and $\epsilon_{\infty}$ are continuous functions. Next, note that:
    \begin{equation}
            \norm{\mb{a}(\mb{x}') - \mb{a}(\mb{x})}_2 \leq \epsilon_L(\mb{x}, \mb{x}') L_{\mb{A}} + \epsilon_{\infty}(\mb{x}, \mb{x}') \norm{\mb{A}}_{\infty}.
    \end{equation}
    Similarly,
    \begin{equation}
        \vert b(\mb{x}') - b(\mb{x})\vert_2 \leq \epsilon_L(\mb{x}, \mb{x}') L_{\mb{b}} + \epsilon_{\infty}(\mb{x}, \mb{x}') \norm{\mb{b}}_{\infty}.
    \end{equation}
    Therefore,
    \begin{align}\label{eqn:upperbound}
        \vert \mb{a}(\mb{x})^\top\mb{u}' + b(\mb{x}) \vert &\leq \ell(\mb{x}', \mb{u}') + \epsilon_L(\mb{x}, \mb{x}')(L_{\mb{A}}\norm{\mb{u}'}_2 + L_{\mb{b}}) \nonumber \\ 
        & \qquad+ \epsilon_{\infty}(\mb{x}, \mb{x}')(\norm{\mb{A}}_{\infty}\norm{\mb{u}'} + \norm{\mb{b}}_{\infty}).
    \end{align}
    
    While $\epsilon_L$ and $\epsilon_{\infty}$ decrease as the test point approaches data points, without estimators as in (\ref{eq:lyapderest}), the observed loss term $\ell$ can remain large. By including such estimators, the observed loss term may be reduced, but the bound must be modified with the following additional continuous function:
    \begin{equation}\label{eqn:esterrterm}
        \epsilon_{\cal{H}}(\mb{x}, \mb{x}', \mb{u}') = \vert (\hat{\mb{a}}(\mb{x}) - \hat{\mb{a}}(\mb{x}'))^\top \mb{u}' +  \hat{b}(\mb{x}) - \hat{b}(\mb{x}') \vert,
    \end{equation}
    which accounts for potential error in the estimation at the test point. $\epsilon$ is then specified as the total upper bound.
\end{proof}

\begin{corollary}
    The uncertainty set $\Delta$ as specified in (\ref{eqn:uncertaintyset}) is continuous with respect to the Hausdorff metric.
\end{corollary}
\begin{proof}
    Note that the inequality constraints in (\ref{eqn:uncertaintyset}) can be rewritten as:
\begin{align}\label{eqn:deltasetmat}
    \Delta(\mb{x}) =&~ \left\{ (\mb{a}, b) \in \R^m \times \R: \bs{\Xi}\begin{bmatrix}\mb{a} \\ b\end{bmatrix}\preceq\bs{\xi}(\mb{x})\right\},
\end{align}
where $\bs{\Xi}\in\R^{2\vert D\vert\times (m+1)}$, $\bs{\xi}:\cal{X}\to\R^{2\vert D\vert}$ and $\preceq$ denotes elementwise inequality. The function $\bs{\xi}$ is continuous since $\epsilon$ is continuous; therefore, the results established in \cite{batson1987combinatorial} show that the point-to-set map $\Delta$ is continuous.
\end{proof}

\section{INTEGRATION WITH LEARNING}
\label{sec:learning}

We now explore the practical interplay between learning and systematic improvement of PSS properties, in particular by decreasing the upper bound in 
\eqref{eqn:upperbound}. By decreasing this bound, the uncertainty set in \eqref{eqn:uncertaintyset} can be made smaller, which in turn can increase the state space region over which PSS properties can be certified and/or achieve reduced degradation (see Corollary \ref{cor:usimprove}). As discussed previously, using PSS rather than ISS enables lower-dimensional learning objectives and upper bounds which can be efficiently evaluated during and after learning.

Learning also offers direct ways to decrease the upper bound in \eqref{eqn:upperbound}. As discussed in Section \ref{sec:modeling}, estimators can be used to reduce the observed loss in (\ref{eqn:obsloss}), which appears directly in the upper bound. We can use supervised learning to train such estimators. One complication is that, using baseline controllers, it may not be possible to collect data in regions we wish to certify PSS properties.  As the distances between a point of interest and previously collected data grow, $\epsilon_L$ and $\epsilon_\infty$ can grow larger, weakening uncertainty bounds at the point of interest. By refining the baseline controller using learned models, the system may be controlled towards these regions of interest.


\subsection{Episodic Learning Framework}

We demonstrate the practicality of PSS by incorporating it into an episodic learning framework based on learning CLF time derivatives \cite{taylor2019episodic}.
Controller improvement is achieved by alternating between executing a controller to gather data and refining estimates of residual uncertainty. 
As data collection and learning progresses, the size of uncertainty sets decreases, enabling stronger PSS certifications for the system.

We briefly describe the DaCLyF (Dataset Aggregation for Control Lyapunov Functions) learning approach from \cite{taylor2019episodic}. 
Let $\cal{H}_{\mb{a}}$ and $\cal{H}_{b}$ be nonlinear estimator classes, and let $\cal{H}$ be the class of estimators of the form:
\begin{equation}
    \hat{\dot{W}}(\mb{x}, \mb{u}) = \hat{\dot{V}}(\mb{x}, \mb{u}) + \hat{\mb{a}}(\mb{x})^\top\mb{u} + \hat{b}(\mb{x}),
\end{equation}
for all $\mb{x} \in \cal{X}$, given estimators $\hat{\mb{a}} \in \cal{H}_{\mb{a}}$ and $\hat{b} \in \cal{H}_b$. Equipped with a loss function $\cal{L}: \R \times \R \to \R_+$ and a dataset $D = \{((\mb{x}_i, \mb{u}_i), \dot{V}_i)\}$ obtained from experiments with a baseline controller, we approximately solve the Empirical Risk Minimization (ERM) problem over the class $\cal{H}$ to update the estimate of the Lyapunov function derivative. Then, the controller is updated with an augmenting controller of the form:
\begin{align}
    \mb{u}'(\mb{x}) = \argmin_{\mb{u}' \in \R^m} ~&~  \frac{1}{2} \begin{bmatrix} \mb{u}(\mb{x})\\ \mb{u}' \end{bmatrix}^\top \mb{P} \begin{bmatrix} \mb{u}(\mb{x})\\ \mb{u}' \end{bmatrix} + \mb{q}^\top \begin{bmatrix} \mb{u}(\mb{x})\\ \mb{u}' \end{bmatrix} + r\nonumber\\
    \mathrm{s.t } ~&~ \hat{\dot{V}}(\mb{x}, \mb{u}(\mb{x}) + \mb{u}') \leq -\alpha(\norm{\mb{x}}) \nonumber\\
    &~ \mb{u}(\mb{x}) + \mb{u}' \in \cal{U},
\end{align}
for all $\mb{x} \in \cal{C}$, with $\mb{P} \in \mathbb{S}^{2m}_+$, $\mb{q} \in \R^{2m}$, $r > 0$, and $\hat{\dot{V}}$ denoting the updated estimator. Here $\mathbb{S}^{2m}_+$ denotes the set of positive semidefinite matrices of size $2m \times 2m$. The augmenting controller, weighted by a trust factor, is additively incorporated with the baseline controller. As more data is collected, the trust factor is increased. This entire procedure is outlined in Algorithm \ref{alg:daclyf}.

The estimator class $\cal{H}$ can include a wide variety of nonlinear functions. Should $\cal{H}_{\mb{a}}$ and $\cal{H}_b$ be classes of Lipschitz continuous estimators, the upper bound (\ref{eqn:esterrterm}) can be weakened further using the associated Lipschitz constants to permit further analysis of the uncertainty function specified in (\ref{eqn:uncertaintyset}). Importantly, this motivates the use of spectrally normalized deep neural networks \cite{miyato2018spectral}, \cite{bartlett2017spectrally} for this estimation problem.

\begin{algorithm}[t]
    \begin{algorithmic}
        \Require Lyapunov function $V$, Lyapunov function derivative estimate $\hat{\dot{V}}_0$, model classes $\cal{H}_{\mb{a}}$ and $\cal{H}_b$, loss function $\cal{L}$, set of initial conditions $\cal{X}_0$, nominal state-feedback controller $\mb{u}_0$, number of experiments $T$, sequence of trust coefficients $0 \leq w_1 \leq \cdots \leq w_T \leq 1$
        \Statex
        \State $D = \emptyset$ \Comment{Initialize dataset}
        \For{$k = 1, \dots, T$}
            \State $\mb{x}_0 \leftarrow \textrm{sample}(\cal{X}_0)$ \Comment{Sample initial condition}
            \State $D_k \leftarrow \textrm{experiment}(\mb{x}_0, \mb{u}_{k-1})$ \Comment{Execute experiment}
            \State $D \leftarrow D \union D_k$ \Comment{Aggregate dataset}
            \State $\hat{\mb{a}}, \hat{b} \leftarrow \textrm{ERM}(\cal{H}_{\mb{a}}, \cal{H}_b, \cal{L}, D, \hat{\dot{V}}_0)$ \Comment{Fit estimators}
            \State $\hat{\dot{V}}_k \leftarrow \hat{\dot{V}}_0 + \hat{\mb{a}}^\top\mb{u} + \hat{b}$ \Comment{Update derivative estimator}
            \State $\mb{u}_{k} \leftarrow \mb{u}_0 + w_k \cdot \textrm{augment}(\mb{u}_0, \hat{\dot{V}}_k)$ \Comment{Update controller}
        \EndFor
        \Statex
        \Return $D, \hat{\dot{V}}_T, \mb{u}_T$
    \end{algorithmic}
    \caption{Dataset Aggregation for Control Lyapunov Functions (DaCLyF) \cite{taylor2019episodic}\label{alg:daclyf}}
\end{algorithm}

\subsection{Simulation Results}
\label{sec:sim}
In this section we apply Algorithm \ref{alg:daclyf} to an inverted pendulum model with parametric uncertainty. The pendulum is modeled as a massless rod with torque input at a fixed base. The true mass and the length are perturbed by up to 30\% of their estimated values. The baseline controller is a linear proportional derivative (PD) controller to track angle and angle rate trajectories. The estimators are chosen from the class of two layer neural networks with 200 hidden units and ReLU nonlinearities, mapping concatenated state and Lyapunov function gradients to $\R^m$ and $\R$. The trust factors are chosen in a sigmoid fashion. Naive exploratory control is introduced as in \cite{taylor2019episodic}, with perturbations chosen uniformly at random, independently in each coordinate, and scaled by 25\% of the norm of the current control input.

A comparison of the baseline controller and final augmented controller demonstrating improved tracking performance is shown in Fig. \ref{fig:trackingcomp}. A comparison of PSS bounds for the model-based QP controller and the final augmented controller is shown in Fig. \ref{fig:heatmaps}, with observed trajectories superimposed. The QP controller is unable to keep the system in regions in which the bound is shown to be small. On the other hand, the augmented controller keeps the system close to the desired trajectory, consistently near training data. The bounds are small along the observed trajectory, in comparison.


\begin{figure}[H]
    \centering
    \vspace{-0.15in}
    \includegraphics[scale=0.5]{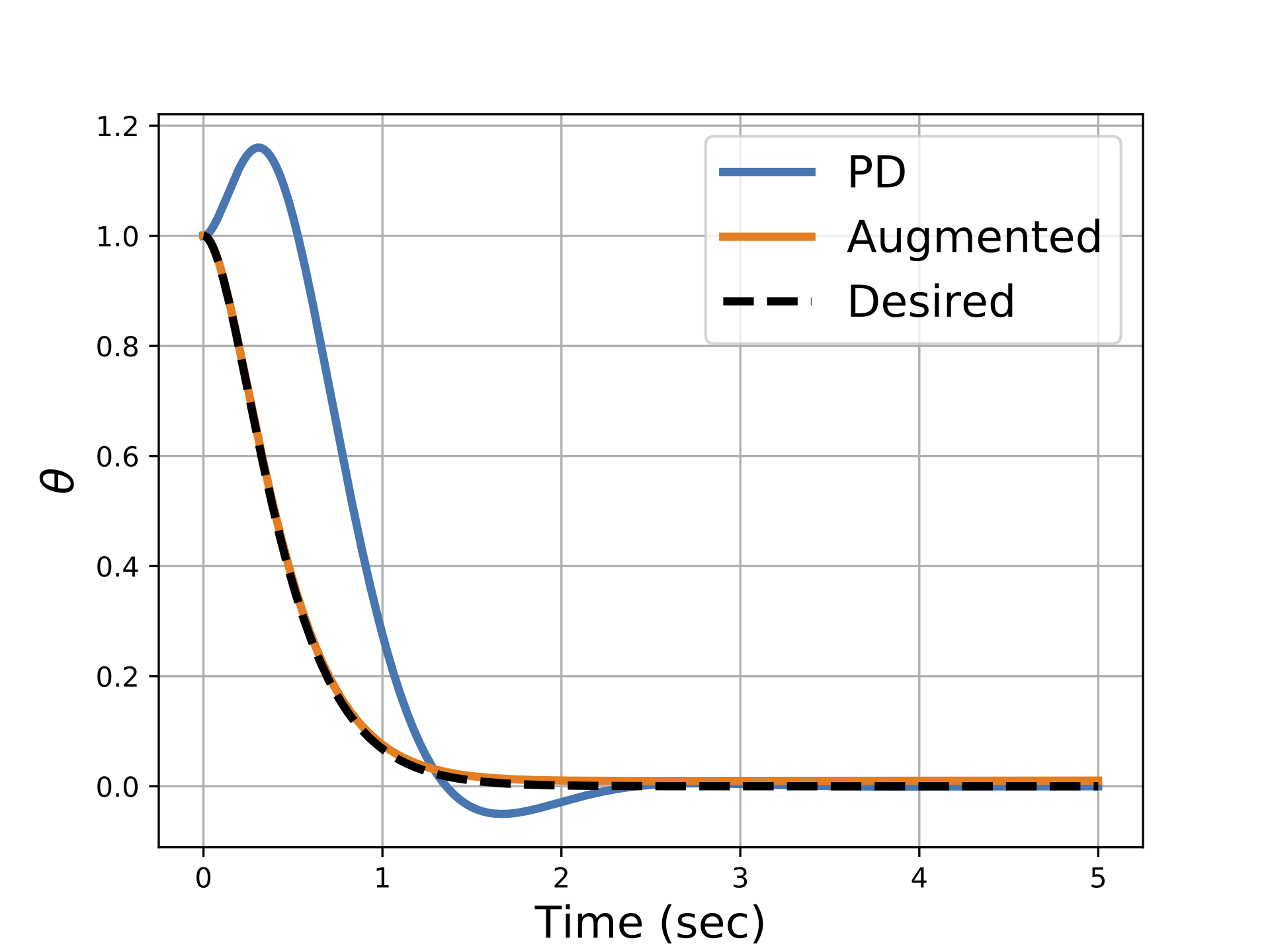}
    \vspace{-0.1in}
    \caption{Comparison of tracking performance for PD controller and final augmented controller. The final augmented controller tracks the desired angle trajectory more effectively.}
    \label{fig:trackingcomp}
\end{figure}

\begin{figure*}[tbph]
    \centering
    \includegraphics[scale=0.5]{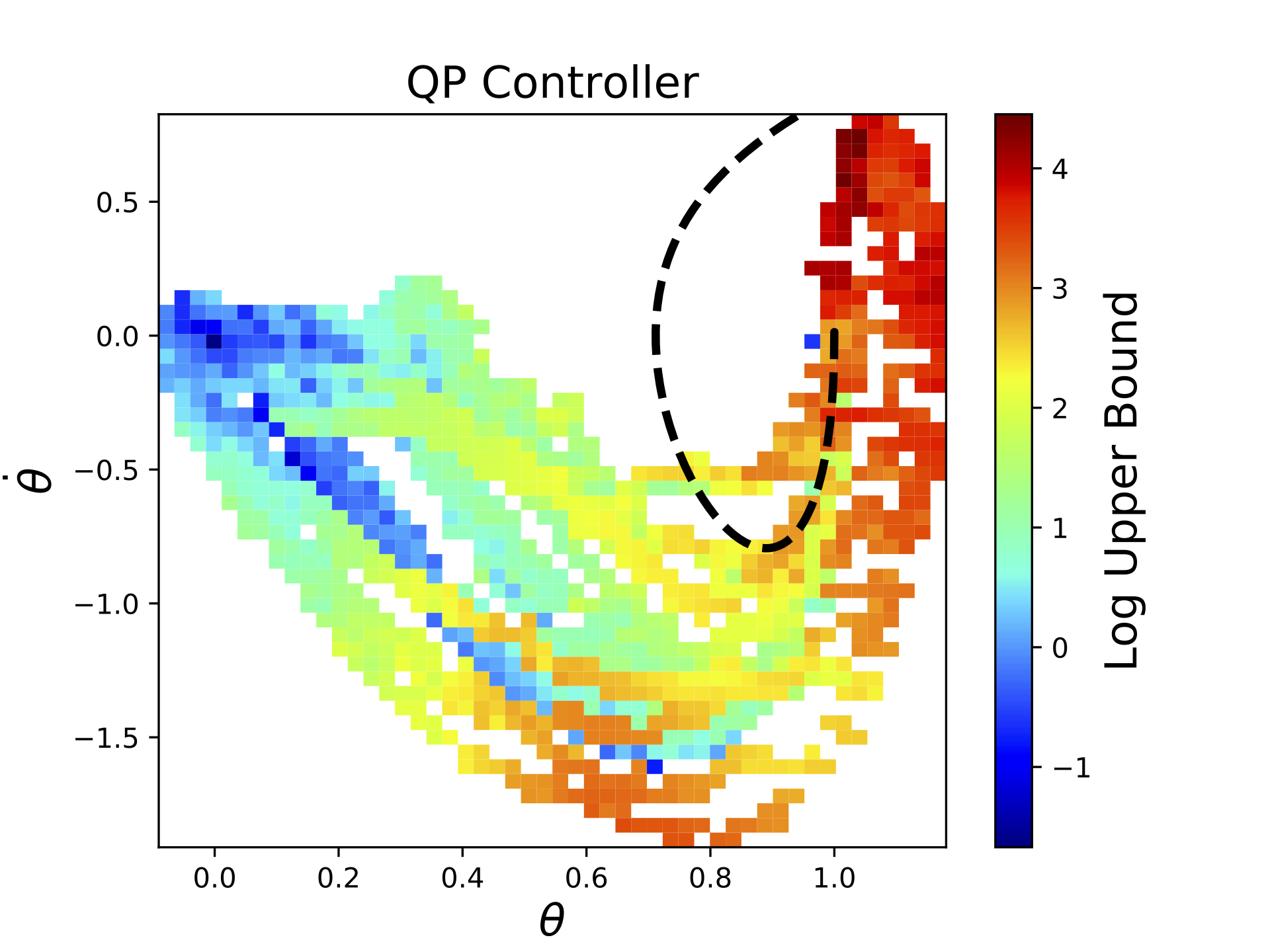}
    \includegraphics[scale=0.5]{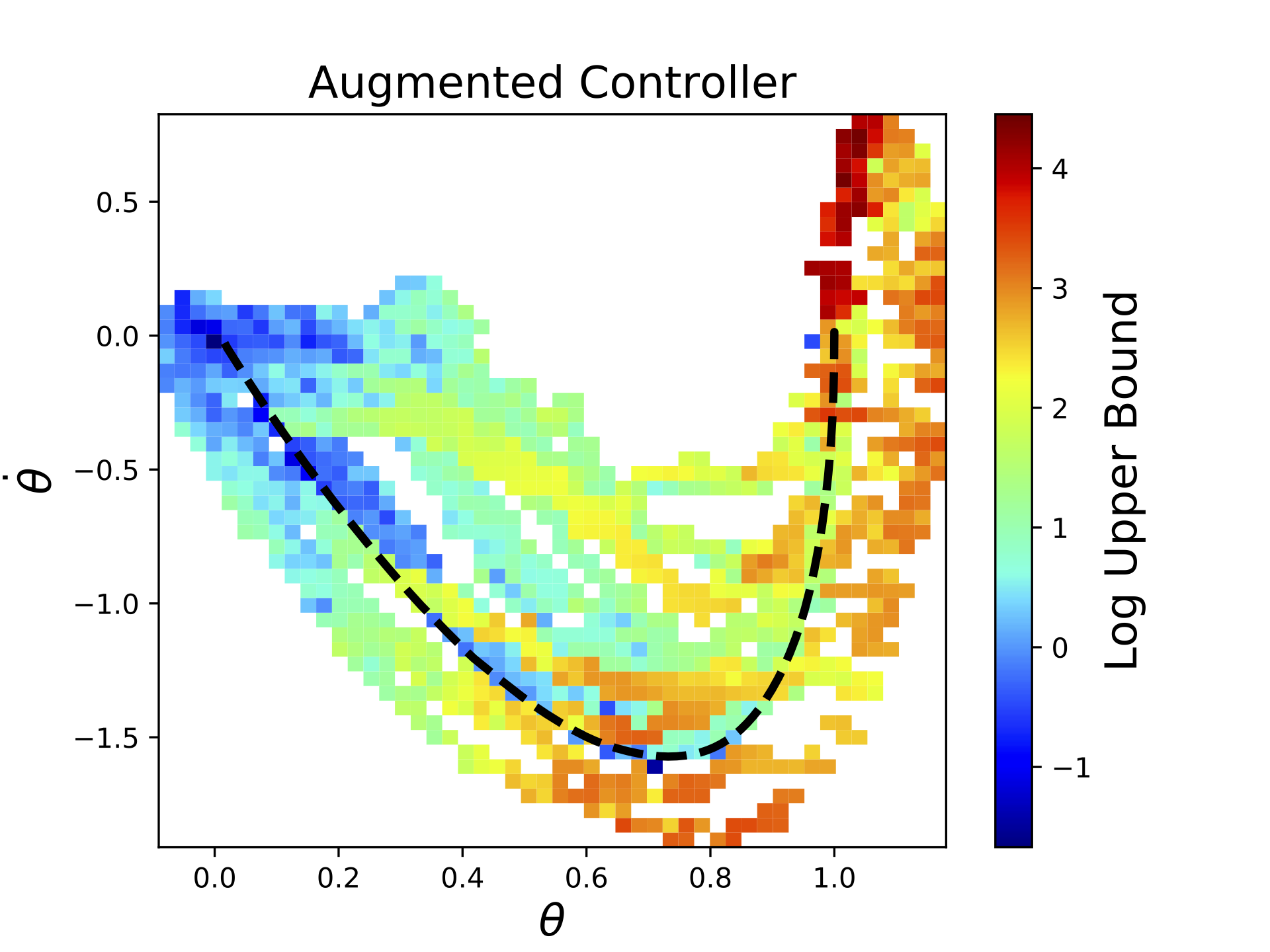}
    \vspace{-0.1in}
    \caption{Comparison of disturbance upper bounds with model based QP controller (left) and final augmented controller (right). The disturbance upper bound is computed from the maximum over uncertainty sets in (\ref{eqn:uncertaintyset}) for both controllers, and observed trajectories are displayed with dashed lines. The augmenting controller keeps the system in regions with lower disturbance bounds while the system leaves the region around the training data under the QP controller. The maps were generated by sampling states randomly about training data points and evaluating the upper bound for each sampled state. The results were then discretized for ease of visualization. Each bin is colored by the maximum disturbance observed in the bin.}
    \label{fig:heatmaps}
\end{figure*}

\section{CONCLUSION}
\label{sec:conclusion}
We presented a novel low-dimensional view of stability for uncertain systems and a method of evaluating PSS behavior using experimental data. This method constructs a bound on disturbances to a CLF derivative, and can be integrated with a machine learning framework to improve PSS behavior. Finally, we validate this procedure on a simulated system.

Future work includes incorporating the upper bounds into online learning settings and developing optimal exploration strategies. Quantifying the impact of learning on PSS provides an objective for deciding how to collect data, also known as the exploration problem in learning literature \cite{moldovan2012safe,aswani2013provably,berkenkamp2016safe,berkenkamp2016bsafe,sui2018stagewise}. In particular, reductions of the uncertainty bound may be used to formulate regret in online learning settings or reward in imitation and reinforcement learning settings.
Additional future work includes extending the notion of PSS to Control Safety and Barrier Functions, more thoroughly studying the benefits of learning low-dimensional representations of the dynamics versus the full-order dynamics, and utilizing PSS to augment controller synthesis for complex real-world robotic systems.


\bibliographystyle{plain}
\bibliography{main}

\end{document}